\documentclass[a4paper, twoside]{report}

\usepackage[english]{babel}
\usepackage[utf8]{inputenc}
\usepackage[T1]{fontenc}

\usepackage[a4paper,top=3cm,bottom=2cm,left=3cm,right=3cm,marginparwidth=1.75cm]{geometry}

\usepackage[backend=biber,style=imperialharvard]{biblatex}
\usepackage{afterpage}
\usepackage{amsmath}
\usepackage{amsthm}
\usepackage{amssymb}
\usepackage{csquotes}
\usepackage{enumitem}
\usepackage{graphicx}
\usepackage{lipsum}
\usepackage{booktabs}
\usepackage[utf8]{inputenc}
\usepackage[english]{babel}

\newtheorem{thm}{Theorem}[section]
\newtheorem{defn}[thm]{Definition}
\newtheorem*{remark}{Remark}

\usepackage{listings}
\usepackage[labelfont=bf,justification=centering]{caption}

\usepackage[ruled,vlined]{algorithm2e}
\makeatletter
\renewcommand{\SetKwInOut}[2]{%
  \sbox\algocf@inoutbox{\KwSty{#2}\algocf@typo:}%
  \expandafter\ifx\csname InOutSizeDefined\endcsname\relax
    \newcommand\InOutSizeDefined{}\setlength{\inoutsize}{\wd\algocf@inoutbox}%
    \sbox\algocf@inoutbox{\parbox[t]{\inoutsize}{\KwSty{#2}\algocf@typo:\hfill}~}\setlength{\inoutindent}{\wd\algocf@inoutbox}%
  \else
    \ifdim\wd\algocf@inoutbox>\inoutsize%
    \setlength{\inoutsize}{\wd\algocf@inoutbox}%
    \sbox\algocf@inoutbox{\parbox[t]{\inoutsize}{\KwSty{#2}\algocf@typo:\hfill}~}\setlength{\inoutindent}{\wd\algocf@inoutbox}%
    \fi%
  \fi
  \algocf@newcommand{#1}[1]{%
    \ifthenelse{\boolean{algocf@inoutnumbered}}{\relax}{\everypar={\relax}}%
    {\let\\\algocf@newinout\hangindent=\inoutindent\hangafter=1\parbox[t]{\inoutsize}{\KwSty{#2}\algocf@typo:\hfill}~##1\par}%
    \algocf@linesnumbered
  }}%
\makeatother

\usepackage{bm}

\usepackage{fancyhdr}
\usepackage{tipa}
\usepackage{tcolorbox}
\definecolor{mdgrey}{rgb}{0.8, 0.8, 0.8}
\usepackage[framemethod=tikz]{mdframed}
\usepackage{lipsum}
\newtheoremstyle{defi}
  {\topsep}%
  {\topsep}%
  {\normalfont}%
  {}%
  {\bfseries}%
  {:}%
  {.5em}%
  {\thmname{#1}\thmnote{~(#3)}}%
\theoremstyle{defi}
\newmdtheoremenv{definitioni}{Definition}
\newmdtheoremenv[
hidealllines=true,
leftline=true,
innertopmargin=0pt,
innerbottommargin=0pt,
linewidth=4pt,
linecolor=gray!40,
innerrightmargin=0pt,
]{definitionii}{Definition}
\newmdtheoremenv[
roundcorner=5pt,
innertopmargin=0pt,
innerbottommargin=5pt,
linewidth=4pt,
linecolor=gray!40,
]{definitioniii}{Definition}

\usepackage[colorinlistoftodos]{todonotes}
\usepackage[colorlinks=true, allcolors=blue]{hyperref}

\title{Is your vote truly secret?\\ Ballot Secrecy iff Ballot Independence: Proving necessary conditions and analysing case studies}

\author{Joe Bloggs}

\bibliography{bibs/bibliography.bib}

\begin{document}
\begin{titlepage}

\newcommand{\HRule}{\rule{\linewidth}{0.5mm}} 


\includegraphics[width=8cm]{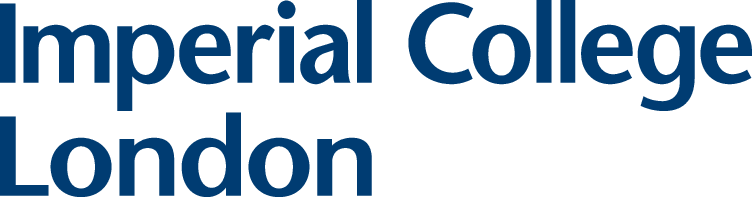}\\[1cm] 
 

\center 


\textsc{\LARGE MEng Individual Project}\\[1.5cm] 
\textsc{\Large Imperial College London}\\[0.5cm] 
\textsc{\large Dyson School of Design Engineering}\\[0.5cm] 

\makeatletter
\HRule \\[0.6cm]
{ \huge \bfseries \@title}\\[0.6cm] 
\HRule \\[1.5cm]
 

\begin{flushleft} \large
\emph{Author:}\\
Aida Maria Manzano Kharman \\[1.2em] 
\end{flushleft}

~
\begin{flushleft} \large
\emph{Supervisors:} \\

Dr. Freddie Page\\ 
\textit{Dyson School of Design Engineering, University of Imperial College London, UK.}\\[1.2em]

Dr. Ben Smyth\\ 
\textit{School of Computer Science, University of Birmingham, UK.}\\[1.2em]

\end{flushleft}
\BlankLine
\BlankLine
\makeatother



{\large \today}\\[1cm] 
\textbf{Word Count}
5000 words

\end{titlepage}

\begin{abstract}
We formalise definitions of ballot secrecy and ballot independence by Smyth, JCS'21 as indistinguishability games in the computational model of security. These definitions improve upon Smyth, draft '21 to consider a wider class of voting systems. Both Smyth, JCS'21 and Smyth, draft '21 improve on earlier works (such as: \cite{29, 71, 28, 96, 97}) by considering a more realistic adversary model wherein they have access to the ballot collection. We prove that ballot secrecy implies ballot independence. We say ballot independence holds if a system has non-malleable ballots. We construct games for ballot secrecy and non-malleability and show that voting schemes with malleable ballots do not preserve ballot secrecy. We demonstrate that Helios does not satisfy our definition of ballot secrecy. Furthermore, the Python framework we constructed for our case study shows that if an attack exists against non-malleability, this attack can be used to break ballot secrecy.
\end{abstract}
\renewcommand{\abstractname}{Acknowledgements}
\begin{abstract}
This work would have not been possible without the kind supervision of Dr Freddie Page, who guided me through unknown paths of varying research topics and embarked on a journey of discovery with me. I would like to thank him for always asking thought provoking questions and transmitting his excitement about learning and research to my project.

I would also like to thank my parents, who have supported me endlessly and lovingly throughout the ups and downs of not only this project, but my entire life. Their unconditional love has kept me going even when I was exhausted. I would not be where I am today without them. Both inspire me to be curious and ambitious, and value learning, research and hard work. Os quiero mucho.

I would also like to thank my dear friend Barty, who has been my moral and technical support throughout the entirety of this year. Thank you for always making me smile, for the laughter and good times we've had, and for being the closest thing to a human version of Stack Overflow that I know. 

Finally, I would like to wholeheartedly thank Dr Ben Smyth. I stumbled across Ben's work online, and he was one of the few researchers that actively answered my questions and showed genuine interest, when he had no reason to do so. Ben took me on as a novice, inexperienced student in the field of cryptography and for reasons that are still unknown to me, has exhibited superhuman levels of patience and kindness towards me. This project would not exist without him, he has been pivotal in everything that I present here, and I am incredibly thankful and humbled by his guidance and support. He has pushed me to challenge myself one step further every time, encouraged me to keep going and has also been brutally honest with me when necessary. This has been the most challenging yet satisfying project I have embarked on as an undergraduate, and simply taking the time to engage with his research has made me gape in awe at the quality of his work. All I can say is thank you for sharing your knowledge with me and helping make this project a reality.  

\end{abstract}

\tableofcontents


\chapter{Introduction}
\newcommand{\BS}{\fontfamily{lmss}\selectfont Ballot Secrecy}
\newcommand{\BSA}{\fontfamily{lmss}\selectfont BS-A}
\newcommand{\BSC}{\fontfamily{lmss}\selectfont BS-C}
\newcommand{\NM}{\fontfamily{lmss}\selectfont Non-Malleability}
\newcommand{\ES}{\fontfamily{lmss}\selectfont Election Scheme}
\newcommand{\NMB}{\fontfamily{lmss}\selectfont Non-Malleable Ballots}

\section{Background}
\label{intro}

\begin{definitionii}
\textbf{\emph{Democracy}}

\noindent \emph{A government in which the supreme power is vested in the people and exercised by them directly or indirectly through a system of representation usually involving periodically held free elections.}
 
\hfill -- Merriam Webster English Dictionary.
\end{definitionii}

Elections are the cornerstone of democracy, they're the expression of people's will and the outcomes determine the members of the legislative power. It is here that the people's voice decides who will rule them. 
Needless to say, the stakes are high. Elections must be fair and free: freedom enables expression of free will and fairness means that election outcomes correctly express voters' votes. However, in the words of \cite{benaloh2015endtoend}: \begin{quote} "Getting the election outcome right isn’t good enough.  Voters deserve convincing evidence that the outcome is correct." \end{quote}

Ensuring that everyone who voted was a member of the electorate and that the results are correct is straightforward when the voting is public: everyone present can verify the results. However, this comes at the cost of voter privacy. Voters can be threatened or bribed to vote a specific way when their vote is public. Because of this, the privacy of vote is recognized as a human right as per Article 21 of the Universal Declaration of Human Rights (\cite{15}). 

The challenge lies in achieving verifiability whilst preserving privacy. Verifiability requires evidence of correct computations, whether this be of casting or recording votes correctly, or of tallying the votes appropriately. End to end (E2E) verifiable voting requires that all three actions are verifiable: casting a vote as intended, recording the vote as cast and tallying the votes as recorded (\cite{benaloh2015endtoend}).

Privacy demands that a voter's vote remains secret. To achieve both E2E verifiability and preserve privacy, a voting scheme must not compromise a voter's privacy whilst convincingly demonstrate to them that their vote was cast as intended, recorded as cast and tallied as recorded. Voters must be able to verify all three steps without being able to prove how they voted to a third party.

Regarding vote privacy, researchers have narrowed down the notion into three crucial properties: (\cite{ali2016overview})
\begin{itemize}
    \item Ballot Secrecy: A voter’s vote is never revealed to anyone.
    \item Receipt-Freeness: A voter cannot prove how they voted to a third party.
    \item Coercion-Resistance: the voter should be able to cast a vote for her intended choice even while appearing to cooperate with a coercer.
\end{itemize}

Coercion resistance implies receipt-freeness, which in turn implies ballot secrecy (\cite{delaune2006coercion}).
\newpage

\section{Purpose}
The project has three goals:
Firstly, mathematically formalise the state-of-the-art definitions in academia of ballot secrecy and ballot independence. Secondly, understand if a relation between the two exists, and if so, mathematically prove or disprove it. Finally, test the results obtained on a case study: do they hold in practical applications?

\section{Structure of the Report}

We explore the academic research on definitions of ballot secrecy and ballot independence and assess the work in the field in Chapter \ref{litreview}.

\noindent We present a generalised definition for both and formalise these definitions using the computational, game-based model of cryptography. Subsequently, we demonstrate that ballot secrecy implies ballot independence, proving that attacks against ballot independence suffice for attacks against ballot secrecy in Chapter \ref{methodology}.

\noindent From known attacks against the seminal Helios voting scheme, we construct a Python framework that showcases how an attack against ballot independence suffices to compromise ballot secrecy in Chapter \ref{results}.

\noindent Finally, we discuss the implications of these results and explore avenues for future work in Chapter \ref{discussion}.

\chapter{Literature Review}
\label{litreview}

\section{Ballot Secrecy}

\cite{92} seminal work kickstarted the discussion of ballot secrecy in academia, whose definition was later developed by \cite{49, 93, 94}. Subsequently, Bernhard et al. consider ballot secrecy definitions for election schemes both with tallying proofs: \cite{25, 26} and without: \cite{27, 28, 29, 71}.

Some definitions compute tallying proofs using the voting scheme’s tallying algorithm: \cite{71}, whereas others use simulators: \cite{29}. Unfortunately, the definitions proposed by Bernhard, Pereira \& Warinschi are too weak, and so is \cite{97}'s subsequent variant of the aforementioned. As shown in \cite{29} section III, these definitions consider voting schemes to be secure which intuitively do not preserve vote privacy. They use simulators for tallying proofs, to task the adversary to distinguish between a real bulletin board and tallying proofs, or a fake bulletin board with simulated tallying proofs. \cite{29} demonstrate that these definitions still consider a system private even if the system reveals exactly how voters voted in the tabulation proof. 

Stronger definitions have been proposed by \cite{27}, which compute tallying proofs using only algorithm tally, but these required weakening. \cite{29} demonstrate that under this definition, any verifiable scheme will be declared not private. Any voting scheme meeting this
definition must allow the authorities to announce any
result that is consistent with the number of votes on the
board. Tally uniqueness and the definition of privacy in \cite{27} are incompatible. This means that the definition cannot be used for schemes that aim to be verifiable.

The innovation that \cite{smyth2021} presents is removing the trust assumption that the communication channels and the bulletin board construction are not compromised. Under \cite{smyth2021}'s definition, the adversary has the power to intercept ballots and construct bulletin boards, thus removing the trust assumptions that previous ballot secrecy definitions relied upon. Prior work focused on identifying vulnerabilities assuming an adversary can coerce a set of voters in the election scheme: \cite{25, 26, 27, 28, 29, 71}. Based on this trust assumption, Smyth \& Bernhard demonstrate that non-malleable ballots are not required for ballot secrecy, however, once this trust assumption is removed, \cite{smyth2021} demonstrates that they are required for ballot secrecy. As such, the latter improves on earlier definitions of ballot secrecy by assuming the adversary is more powerful, and therefore enabling the detection of more vulnerabilities.
\newpage
 
\section{Ballot Independence}

\cite{63} originates the discussion in literature of ballot independence, and the first formalised definition of ballot independence appears in \cite{28} and \cite{27}. The latter mathematically demonstrate that their definitions of ballot secrecy and ballot independence coincide. \cite{smyth2021} also achieves this for his proposed definitions of ballot secrecy and ballot independence. 

Beyond the scope of voting schemes, considerations about independence (\cite{64}) and how the lack thereof can compromise security, have been made in \cite{66}, \cite{67}, \cite{103}, \cite{104}, \cite{105} and \cite{106}. 

\section{Ballot Secrecy and Ballot Independence: What is the link?}
\label{2.3}
There are numerous works that document how the malleability property in encryption schemes can be exploited to compromise ballot secrecy: \cite{98}, \cite{99}, \cite{100} \cite{102}, \cite{smythcortier2011} and \cite{29}. 

\cite{36} introduce the research question as to whether voting systems that allow the casting of meaningfully related ballots uphold ballot secrecy. \cite{107}’s propose a system that achieves this, but \cite{27} disprove this by showing their results do not support their claims (\cite{smyth2021}). \cite{smyth2021} demonstrates that non-malleable ballots are necessary to ensure ballot secrecy, thus answering Bulens, Giry \& Pereira’s research question.

In terms of analysing Helios’ privacy properties, work has been done by \cite{111} showing that Helios variants satisfy earlier definitions of ballot secrecy. As mentioned earlier however, these definitions are too weak, as they fail to detect vulnerabilities in the ballot collection process. \cite{smyth2021} analyses comprehensively the different Helios variants under his newer definitions and demonstrates that the variants with non-malleable ballots (ie: ballot independence) preserve ballot secrecy.

It is worth noting that the existing definitions of ballot secrecy formalise a notion of privacy and thus free will, but this does not mean that they satisfy definitions of verifiability posed earlier. In particular, that ballots are properly formed (ie: cast as intended), recorded (ie:recorded as cast) and correctly tallied (ie: tallied as recorded). However, to ensure free will, one must also consider receipt freeness: \cite{47, delaune2006coercion, 126, 127, 128}, and coercion resistance: \cite{129, 130, 131, 132, 133}.

\BlankLine
\BlankLine

We acknowledge \cite{smyth2021} extensive review of related work in section 8 that informed this literature review.

\chapter{Methodology}
\label{methodology}

\section{Preliminaries: Games and Notation}
It is convention in cryptographic security proofs to construct games that capture security notions as a means to formalise them. These games often entail a series of information exchanges between an honest and a malicious actor. We refer the latter as the adversary, where the adversary is tasked with completing some particular action. If they succeed, they win the game and the security notion is violated. We aid ourselves from the work in \cite{cryptoeprint:2004:332} to understand the techniques used to prove equivalence between two games, and follow the same preliminaries set up in \cite{smyth2021} section 2. After all, we are constructing games for newer versions of the definitions first presented in \cite{smyth2021}'s work, so we preserve the syntax they introduce. 

\begin{itemize}
    \item A game captures a series of information exchanges between an honest actor (an election administrator in this case), a cryptographic scheme (the voting system) and an adversary.
    \item To win, the adversary must perform a task that compromises the security of the cryptographic scheme. 
    \item The tasks an adversary must perform demonstrate \textit{indistinguishability} or \textit{reachability}. For example, {\BS} holds if an adversary cannot distinguish between two distinct voting scenarios that they have observed.
    \item The games are probablistic and output booleans. We consider computational security as opposed to information-theoretic security. Attacks by adversaries in non-polynomial time and with negligible success are not considered, since such breaks are infeasible in practice.
    \item We let $\mathnormal{x} \gets \mathnormal{Y}$ denote the assignment of $\mathnormal{Y}$ to $\mathnormal{x}$, and $x \gets _R\{0, 1\}$ denote the assignment of $\mathnormal{x}$ to an element of the set $\{0,1\}$ chosen uniformly at random. 
    \item Adversaries are \textit{stateful}, meaning they have access to information that they receive throughout the entirety of the game. \item Adversaries win a game by causing it to output True ($\top$).
    \item We denote the probability of an adversary winning the game {\fontfamily{lmss}\selectfont Game} as \textnormal{\fontfamily{lmss}\selectfont Succ(Game)}.

\end{itemize}
\newpage
\section{Election Scheme Syntax}
We construct definitions of ballot secrecy and ballot independence for voting schemes that perform the following operations: 
The administrator generates a keypair, each voter constructs a ballot for their vote, the administrator tallies the set of cast ballots and recovers the election outcome from the aforementioned tally. 
It is necessary to formalise these voting systems in way that allows us to model them in cryptographic games for ballot secrecy and ballot independence. As such, we define an election scheme below:
\begin{defn}
\label{es}
An \textnormal{\ES} is a tuple of probabilistic polynomial-time algorithms\\ $(\mathrm{Setup, Vote, Partial-Tally, Recover})$ such that:

\noindent $\mathrm{Setup}$ denoted $(pk, sk) \gets \mathrm{Setup}(k)$ is run by the administrator. The algorithm takes security parameter $k$ as an input, and returns public key $pk$ and private key $sk$.\;
\BlankLine
\noindent $\mathrm{Vote}$ denoted $b \gets \mathrm{Vote}(pk, v, k)$ is run by the voters. The algorithm takes public key $pk$, the voter's vote $v$ and security parameter $k$ as inputs and returns a ballot $b$, or an error ($\perp$).\;
\BlankLine
\noindent $\mathrm{Partial-Tally}$ denoted $e \gets \mathrm{Partial-Tally}(sk, \mathfrak{bb}, k)$ is run by the administrator. The algorithm takes secret key $sk$, bulletin board $\mathfrak{bb}$, and security parameter $k$ as inputs and returns evidence $e$ of a computed partial tally.\;
\BlankLine
\noindent $\mathrm{Recover}$ denoted $\mathfrak{v} \gets \mathrm{Recover}(\mathfrak{bb}, e, pk)$ is run by the administrator. The algorithm takes bulletin board $\mathfrak{bb}$, evidence $e$ and public key $pk$ as inputs and returns the election outcome $\mathfrak{v}$.
\end{defn}

\section{Formalising Privacy Notions}
We introduce the informal definitions that \cite{votetech} presents and subsequently provide the formalised definitions that we constructed.

\subsection{Ballot Secrecy}
\cite{votetech} informally defines ballot secrecy as follows:
\begin{quote}
A voting system guarantees ballot secrecy if an
adversary corrupting all but one administrator cannot determine the value of
bit $\beta$ from a bulletin board they compute using voters’ ballots — each ballot expressing vote $v_\beta$, computed upon the adversary’s demand for a ballot expressing vote $v_0$ or $v_1$ — along with the non-corrupt administrator’s public key share
and evidence, wherein for bulletin-board ballots computed on demand with inputs
$(v_{1,0}, v_{1,1}), . . . ,(v_{n,0}, v_{n,1})$, we have $v_{1,0}, . . . , v_{n,0}$ is a permutation of $v_{1,1}, . . . , v_{n,1}$.

\end{quote}

From the definition above, we introduce a game {\BS} (below) between a challenger and an adversary that proceeds as follows: the challenger generates a key pair (Line 1), and flips a coin (Line 2). A set is initialised to keep track of the challenge ballots ($b$) (Line 3). The adversary computes a bulletin board ($\mathfrak{bb}$) from ballots expressing non-corrupt voters' vote (Line 4), where $\mathcal{A^O}$ denotes adversary $\mathcal{A}$'s access to Oracle $\mathcal{O}$. Upon the adversary's demand, the oracle $\mathcal{O}$ constructs a ballot expressing a vote $v_0$  or $v_1$, where $v_0, v_1$ are two distinct voting instances. Both must belong to the set $V$, of all the candidates a voter can vote for. The challenge ballot is then added to the set $L$. The adversary can also include ballots they construct (representing corrupt voters' votes) on the bulletin board. The challenger then performs a partial tally of the bulletin board using their private key and outputs the resulting evidence ($e$) (Line 5). The adversary must then determine the value of the coin (bit $\beta$) from all public outputs (Line 6).
\newpage
\begin{defn}[\textbf{\BS}]
Let $\Gamma$ = $(\mathrm{Setup}, \mathrm{Vote}, \mathrm{Partial-Tally}, \mathrm{Recover})$ be an election scheme, $\mathcal{A}$ be an adversary, $V$ be the set of all possible candidates, $k$ be a security parameter and \textnormal{\BS} the following game:

\LinesNumbered
\begin{algorithm*}[H]

\SetKw{Return}{return}

$\mathnormal{pk, sk} \gets \mathrm{Setup}(k) $\;
$\mathnormal{\beta}$ $\gets$ $\mathnormal{_R\{0, 1\}}$\;
$L \gets \emptyset$\;
$\mathfrak{bb}$ $\gets$ $\mathcal{A^O}(pk,k)$\;
$\mathnormal{e}$ $\gets$ $\mathrm{Partial-Tally}(sk, \mathfrak{bb}, k)$\;
$\mathnormal{g}$ $\gets$ $\mathcal{A}(\mathfrak{bb}, e)$\;

\BlankLine
\Return $g = \beta$ $\wedge$ $balanced(\mathfrak{bb}, L)$ $\wedge$ $v_0, v_1 \in V$\;
\BlankLine

\caption{Ballot Secrecy($\Gamma$, $\mathcal{A}$, $V$, $k$)}
\end{algorithm*}

\label{algorithm:posit}

Predicate $balanced$ and oracle $\mathcal{O}$ are defined as follows:

\begin{itemize}
    \item Predicate $balanced(\mathfrak{bb}, V, L)$ holds if for all votes $v \in V$ we have $| \{ b \mid b \in \mathfrak{bb} \wedge \exists v_1 . (b, v, v_1) \in L \}| = | \{ b \mid b \in \mathfrak{bb} \wedge \exists v_0 . (b, v_0, v) \in L \}|$.
    \item Oracle $\mathcal{O}(v_0, v_1)$ computes $b$ $\gets$ $\mathrm{Vote}(pk, v_\beta, k)$; $L \gets L \cup \{(b, v_0, v_1)\}$
    and outputs $b$, where $v_0, v_1 \in V$.
\end{itemize}

We say $\Gamma$ satisfies \textnormal{\BS} if for all probabilistic polynomial-time adversaries $\mathcal{A}$ and all sets of candidates $V$, there exists a negligible function {\fontfamily{lmss}\selectfont negl}, such that for all security parameters $k$, we have:
\textnormal{\fontfamily{lmss}\selectfont Succ(Ballot Secrecy($\Gamma$, $\mathcal{A}$, $V$, $k$)) $\leq$ $\frac{1}{2} + $ {\fontfamily{lmss}\selectfont negl}($k$)}.
\end{defn}

This definition tasks the adversary with distinguishing between two distinct voting scenarios that they have observed, where one voting scenario is a permutation of the other. If the adversary fails to distinguish them from the information they have access to, then {\BS} is preserved. Intuitively, if the adversary can win the game, then there exists a way of distinguishing ballots. 

To avoid trivial distinctions being considered privacy breaches, we introduce the $balanced$ condition. This states that the number of challenge ballots for $v_0$ in the bulletin board must be the same as the number of challenge ballots for $v_1$, where $v_0$ and $v_1$ are the left and right inputs to the oracle respectively, and are provided to the oracle by the adversary. Without this condition, the adversary could query the oracle with inputs $(2,1)$ and compute a bulletin board with the outputted ballots. From the tally of the bulletin board, the adversary could easily deduce the bit $\beta$.

For example: if $v_0$ represents a vote for Labour and $v_1$ represents a vote for Conservative, the adversary constructs a bulletin board and adds two adversary constructed ballots: one for Labour and one for Conservative. Then the adversary queries the Oracle: $b \gets \mathcal{O}(Labour, Conservative)$ and receives a ballot $b$ that contains Labour if $\beta = 0$ or Conservative if $\beta = 1$. The adversary adds the challenge ballot to the bulletin board and a tally is computed. The tally returns the following result: Labour 2, Conservative 1. The adversary therefore knows that $\beta = 0$ because the challenge ballot contained a vote for Labour. However, under our definition of {\BS} the adversary loses the game because the bulletin board is not balanced. There is one challenge ballot for Labour but no challenge ballots for Conservative.

\newpage

\subsection{Ballot Independence}
Ballot independence means that a voter cannot compute a meaningfully related ballot by observing another voter's ballot. To do so, the ballots would have to be malleable, meaning one can be transformed into another. We say that a system preserves ballot independence if it has non-malleable ballots.
\cite{votetech} informally defines non-malleable ballots as follows:
\begin{quote}
\label{nm}
A voting system has non-malleable ballots, if an adversary corrupting all but one administrator cannot determine the value of bit $\beta$ from that administrator’s public key and the election result for a bulletin board the adversary computes using a challenge ballot, namely, a ballot expressing vote $v_\beta$, computed upon the adversary’s demand for a ballot expressing vote $v_0$ or $v_1$, wherein the challenge ballot does not appear on the bulletin board and the election result is derived from evidence computed by the non-corrupt administrator and the adversary.

\end{quote}

From definition the definition above we introduce a formal game definition for {\NM} (below) between a challenger and an adversary that proceeds as follows:
Lines 1 and 2 are identical to {\BS}. The adversary returns the two possible votes the challenge ballot can express a vote for (Line 3) and the challenger then constructs a challenge ballot containing a vote $v_\beta$ (Line 4). The adversary computes a bulletin board whilst having access to the challenge ballot and the public key share (Line 5). Line 6 is equivalent to Line 5 in {\BS}. The challenger computes the election outcome ($\mathfrak{v}$) from the bulletin board, the evidence and the public key (Line 7). The outcome is given to the adversary, which must then determine the value of the coin (bit $\beta$) (Line 8), wherein the challenge ballot $\mathnormal{b}$ does not appear in the bulletin board (Line 9).

\begin{defn}[\textbf{\NM}]
Let $\Gamma$ = $(\mathrm{Setup}, \mathrm{Vote}, \mathrm{Partial-Tally}, \mathrm{Recover})$ be an election scheme, $\mathcal{A}$ be an adversary, $V$ be the set of all possible candidates, $k$ be a security parameter and \textnormal{\NM} the following game:

\LinesNumbered
\begin{algorithm*}[H]
\SetAlgoLined
\SetKwInOut{KwInput}{Input}
\SetKwInOut{KwOutput}{Output}
\SetKwInOut{KwPre}{Pre}
\SetKw{Return}{return}

$\mathnormal{pk, sk} \gets \mathrm{Setup}(k) $\;
$\mathnormal{\beta}$ $\gets$ $\mathnormal{_R\{0, 1\}}$\;
$v_0, v_1 \gets \mathcal{A}(pk, k)$\;
$\mathnormal{b}$ $\gets$ $\mathrm{Vote}(pk, v_\beta, k) $ \; 
$\mathfrak{bb}$ $\gets$ $\mathcal{A}(b)$\;
$e \gets \mathrm{Partial-Tally}(sk, \mathfrak{bb}, k)$\;
$\mathfrak{v} \gets \mathrm{Recover}(\mathfrak{bb}, e, pk)$\;
$g \gets \mathcal{A}(\mathfrak{v})$\;

\BlankLine
\Return $g = \beta \wedge b \notin bb \wedge v_0, v_1 \in V$\;
\BlankLine

\caption{Non-Malleability($\Gamma$, $\mathcal{A}$, $V$, k)}
 
\end{algorithm*}
We say $\Gamma$ satisfies \textnormal{\NM} if for all probabilistic polynomial-time adversaries $\mathcal{A}$ and all sets of candidates $V$, there exists a negligible function {\fontfamily{lmss}\selectfont negl}, such that for all security parameters $k$, we have:
\textnormal{\fontfamily{lmss}\selectfont Succ(Non-Malleability($\Gamma$, $\mathcal{A}$, $V$, $k$)) $\leq$ $\frac{1}{2} + $ {\fontfamily{lmss}\selectfont negl}($k$).}

\end{defn}
Game {\NM} is satisfied if the adversary cannot know what vote ($v_0$ or $v_1$) the challenge ballot $b$ is for. To avoid trivial distinctions, we require that the challenge ballot does not appear on the bulletin board. Otherwise the adversary could compute a bulletin board with adversarial ballots, add the challenge ballot to the bulletin board, and from the tally determine what the challenge ballot contained.  
If the ballots are malleable, then the adversary can exploit this property and derive a meaningfully related ballot $b'$ from the challenge ballot $b$. They can then add $b'$ to the bulletin board and from the tally determine what the $b$ contained, given that the adversary knows the relation between the $b$ and $b'$.

In our definition, we task an adversary with transforming a challenge ballot into a meaningfully related ballot. If the adversary succeeds, we know the voting scheme does not satisfy {\NM}.

\subsection{Comparing Ballot Secrecy and Ballot Independence}

We compare our games for {\BS} and {\NM} and we can find the following three main distinctions:
\begin{enumerate}
    \item In {\NM}, the adversary only receives one challenge ballot, whereas in {\BS} the adversary can receive an arbitrary number of challenge ballots through the oracle calls.
    \item The adversary in {\BS} can add both challenge ballots and adversary constructed ballots to the bulletin board that they construct, whereas the adversary in {\NM} can only add adversary constructed ballots.
    \item The bulletin board must be $balanced$ in {\BS}, whereas in {\NM} the bulletin board must not contain any challenge ballots ($b \notin \mathfrak{bb}$). Indeed, if the bulletin board in {\NM} satisfies $b \notin \mathfrak{bb}$, then it will be balanced, since the number of challenge ballots for $v_0$ and for $v_1$ will both be 0.
    \item The adversary in {\BS} has access to the evidence provided by the administrator upon completing the partial tally, whereas the adversary in {\NM} does not. They have access to the election outcome $\mathfrak{v}$.
\end{enumerate}
\chapter{Results}
\label{results}
\section{Proving Ballot Secrecy implies Ballot Independence}
As described in section \ref{2.3}, earlier work suggests a link between ballot secrecy and ballot independence (ie: non-malleability) exists. Having formalised the state-of-the-art definitions for ballot secrecy and non-malleability, we propose theorem \ref{theorem1}. 

\begin{thm}[\BS \ $\Rightarrow$ \NM]
\label{theorem1}
Given an election scheme $\Gamma$ satisfying \textnormal{\BS}, \textnormal{\NM} is satisfied. 
\end{thm}

\begin{proof}
\label{proof1}
Suppose $\Gamma$ does not satisfy {\NM}, where $\Gamma = (\mathrm{Setup}, \mathrm{Vote}, \mathrm{Partial-Tally}, \mathrm{Recover})$. There exists a probabilistic polynomial-time adversary $\mathcal{A}$, such that for all candidates in set of candidates $V$ and all negligible functions {\fontfamily{lmss}\selectfont negl}, there exists a security parameter $k$ such that {\fontfamily{lmss}\selectfont Succ(Non Malleability($\Gamma$, $\mathcal{A}$, $V$, $k$)) $>$ $\frac{1}{2} + $ {\fontfamily{lmss}\selectfont negl}($k$)}. We construct an adversary $\mathcal{B}$ against {\BS} \ from $\mathcal{A}$. 
\begin{itemize}
    \item $\mathcal{B}(pk, k)$ computes $v_0, v_1 \gets \mathcal{A}(pk, k)$ and outputs ($v_0, v_1$).
    \item $\mathcal{B}()$ computes $b \gets \mathcal{O}(v_0, v_1)$; $\mathfrak{bb} \gets \mathcal{A}(b)$ and outputs $\mathfrak{bb}$.
   \item $\mathcal{B}(\mathfrak{bb}, e)$ computes $\mathfrak{v} \gets \mathrm{Recover}(\mathfrak{bb}, e, pk)$; $g \gets \mathcal{A}(\mathfrak{v})$ and outputs $g$.

\end{itemize}

Adversary $\mathcal{B}$ simulates $\mathcal{A}$'s challenger to $\mathcal{A}$. The bulletin board is constructed by $\mathcal{A}$ in {\NM}, meaning it does not contain any challenge ballots. This implies that the bulletin board is balanced, as the set of challenge ballots for $v_0$ and the set of challenge ballots for $v_1$ are both empty. Furthermore, the adversary in {\BS} takes election scheme $\Gamma$ as an input, meaning that it has access to algorithm $\mathrm{Recover}$ and has sufficient information to obtain $\mathfrak{v}$, the election outcome, from variables $\mathfrak{bb}$ and $e$. Thus, $\mathcal{B}$ can obtain $g$ from $\mathcal{A}(\mathfrak{v})$. 
Therefore, adversary $\mathcal{B}$ has the necessary conditions to win \BS. It follows that {\fontfamily{lmss}\selectfont Succ(Non Malleability($\Gamma$, $\mathcal{A}$, $V$, $k$))} = {\fontfamily{lmss}\selectfont Succ(Ballot Secrecy($\Gamma$, $\mathcal{A}$, $V$, $k$))} and {\fontfamily{lmss}\selectfont Succ(Ballot Secrecy($\Gamma$, $\mathcal{A}$, $V$, $k$)) $>$ $\frac{1}{2} + $ {\fontfamily{lmss}\selectfont negl}($k$)}, concluding our proof.
\end{proof}

\newpage

\section{Practical Application of Results: Helios Case Study}
\subsection{A brief primer: homomorphic encryption schemes}

In malleable encryption schemes, a ciphertext can be transformed to another ciphertext, that when decrypted, produces a meaningfully related plaintext to that of the original one. Often this is seen as an undesirable property as it allows attackers that can intercept ciphertexts to change the meaning of the intended message by computing a different ciphertext. We can also appreciate how in the context of voting schemes, this property can be exploited to compromise voter privacy, as shown in our proof above. 

However, some encryption schemes are malleable by design. Homomorphic encryption allows administrators to perform computations on encrypted data without having to decrypt it first. This is useful for privacy preserving applications, where mathematical operations can be done on encrypted data without having to decrypt it. Because of this, homomorphic encryption schemes are sometimes used in electronic voting. They allow for the tallying of encrypted ballots without ever having to decrypt any of them individually.

Helios is a well known example of an E2E verifiable, open source, web-based electronic voting scheme that uses homomorphic encryption. It has been used since 2009 by KU Leuven to elect their president (\cite{adida2009electing}) and by Princeton to elect their student representatives and since 2010 by the International Association for Cryptographic Research to elect board members (\cite{haber2010helios}). Remarkably, more than 2 million votes have been cast using Helios.

This voting system uses a variant of El-Gamal encryption that is additively homomorphic (\cite{adida2008helios}). A basic example of an additively homomorphic scheme would unfold as follows: 

Aida encrypts a message $m_1$, computes $c_1 \gets E(m_1;r_1)$ where $r_1$ is a random coin, and outputs ciphertext of the message $c_1$. Ben encrypts a message $m_2$, computes $c_2 \gets E(m_2; r_2)$, where $r_2$ is another coin and outputs ciphertext of the message $c_2$. To find the sum of both messages, Freddie computes $c_3 = c_1 \oplus c_2$ and outputs $c_3$. Given the system is additively homomorphic, $c_3 = E(m_1 + m_2; r_1 \otimes r_2)$ holds. This allows Freddie to find the sum of both messages by decrypting $c_3$  without knowing what Aida and Ben's individual messages were.  In summary:
$E(m_1) + E(m_2) = E(m_1 + m_2)$
holds for an additively homomorphic encryption scheme.

Attacks against Helios exploiting the malleability property have been documented in \cite{smythcortier2011} and show how ballot secrecy can be compromised. Indeed, if the attacker has access to the ballot collection stage, they can intercept a ballot, compute a meaningfully related ballot and cast that one to the bulletin board they compute. From the tally of the bulletin board, they can deduce what the original ballot contained, without ever having to decrypt it.

\subsection{Exploiting our proof}
Knowing that Helios has malleable ballots, given that it uses homomorphic encryption, we know that an adversary $\mathcal{A}$ can be constructed against {\NM} to win the game. If our theorem holds true, then we should be able to construct an adversary $\mathcal{B}$ against {\BS} from $\mathcal{A}$ that simulates $\mathcal{A}$'s challenger to win {\BS}.

To test this, we first write a Python script for {\NM} and for {\BS} games. The scripts run the game with a dummy voting system that provides no encryption, against an adversary with no strategy, that always guesses $\beta = 1$. Naturally, the adversary guesses the correct bit 50\% of the time, as $\beta$ can be either 0 or 1. 

Next, we extract from Helios the four algorithms that the voting scheme must perform according to our definition presented in \ref{es}: ($\mathrm{Setup, Vote, Partial-Tally, Recover}$). Helios is open source, so the backend encryption code is publicly available at: \url{https://github.com/benadida/helios}.

Subsequently, we replace the dummy voting scheme class with the Helios class and we construct an adversary $\mathcal{A}$ against {\NM}. The adversary receives a challenge ballot $b$ and exploits the malleability property of Helios to compute a challenge ballot $b'$. 

The Helios ballots are comprised of ciphertexts encrypting the voter's vote, paired with a series of zero knowledge proofs (ZKPs) that demonstrate correctness of ballot formation. There are two sets of ZKPs: individual proofs and an overall proof, and each ZKP is comprised of a challenge, commitment and response value. Individual proofs demonstrate that the ciphertext encrypts a value between 0 and the maximum number of options a voter can select, for all the possible candidate choices. Overall proofs that show that the ballot contains a total number of votes that is valid. Each ballot $b$ is verified before being cast. Upon receiving $b$, the adversary alters the response value of individual proofs by adding $q$. \footnote{In El-Gamal encryption, a large prime $p$ is selected (2048 bits) such that $p = 2q + 1$ where $q$ is also a prime. $p$ and $q$ are components of the public key and used in the encryption operation. For further details on El-Gamal refer to \cite{Elgamal85apublic}} The adversary can readily access $q$ as it is part of the public key, which is given to them in line 3 of {\NM}. 

The proof of correct ballot formation will still hold with the new response values because the verification is carried out using modulo arithmetic. The new ballot $b'$ is distinct to the original challenge ballot $b$, as the individual proofs are different, but the vote that the ballot encrypts will remain unchanged. The adversary can therefore add $b'$ to the bulletin board it without breaking the $b \notin \mathfrak{bb}$ condition. 

The bulletin board now contains adversarial ballots and ballot $b'$. The administrator computes the partial tally of the bulletin board and returns the election outcome to the adversary. The adversary knows what vote all of the ballots in the bulletin board contained except that of $b'$, but deducing that vote is now trivial: the adversary simply computes the difference between the sum of their adversarial ballots and the election outcome. The result reveals what vote $b'$ contained. 

Therefore, we show that there exists a probabilistic polynomial-time adversary $\mathcal{A}$ that can win {\NM} with {\fontfamily{lmss}\selectfont Succ(Non-Malleability($\Gamma$, $\mathcal{A}$, $V$, $k$)) $>$ $\frac{1}{2} + $ {\fontfamily{lmss}\selectfont negl}($k$)} and as such {\NM} does not hold.

Finally, we construct an adversary $\mathcal{B}$ against {\BS} from adversary $\mathcal{A}$ that simulates $\mathcal{A}$'s challenger. For each adversary call in {\BS}, $\mathcal{B}$ performs the exact operations detailed in \ref{proof1}.

$\mathcal{B}$ is able to win {\BS} following these steps with {\fontfamily{lmss}\selectfont Succ(Ballot Secrecy($\Gamma$, $\mathcal{B}$, $V$, $k$) $>$ $\frac{1}{2} + $ {\fontfamily{lmss}\selectfont negl}($k$))}, thus showing that {\BS} does not hold.

The code is available on github at \url{https://github.com/aidamanzano/masters}.

\chapter{Discussion}
\label{discussion}

Our work presents formalised definitions for the ones presented in \cite{votetech}. From these we can quantitatively analyse whether a voting scheme satisfies them or not. Following the research that has previously been conducted in the field, we present a theorem that claims that a relationship between ballot secrecy and ballot independence holds, namely that {\BS} implies {\NM}. We mathematically prove this theorem thereby showing that a voting scheme with malleable ballots cannot preserve ballot secrecy.

The implications of this proof mean that voting systems that use malleable encryption methods are vulnerable to attacks compromising voters' privacy. From this proof we construct a Python framework to test if attacks against {\NM} can be exploited to compromise {\BS} in a real life scenario. Helios is a widely used voting system that has malleable ballots, and we show that an adversary can successfully exploit this property against {\BS}.

Whilst attacks against Helios are known of (we implement the one described in \cite{smythcortier2011}), our work's contribution extends beyond this. Our Python framework  demonstrates practically how an attack against {\NM} can be transformed to an attack against {\BS}. This facilitates testing attacks in voting systems across different privacy definitions. We hope that this contribution can streamline and aid the design and cryptoanalysis of future voting schemes by considering privacy in the design process bottom up as opposed to top down. 

Testing Helios in our Python framework was particularly challenging because Helios was not designed to be tested against {\BS} or {\NM}, it was designed to work. It is the job of cryptoanalysts to detect vulnerabilities once the finished product is deployed, but we disagree with this workflow and hope that testing future voting schemes is easier if they are designed such that the algorithms $\mathrm{Setup, Vote, Partial-Tally, Recover}$ are easily abstractable from the system.

Future lines of work include proving the the following remark holds true:
\begin{remark}[ \NM \ $\Rightarrow$ \BS]
\textit{Given an election scheme} $\Gamma$ \textit{satisfying} {\NM},  {\BS} \ is satisfied, where $\Gamma = (\mathrm{Setup}, \mathrm{Vote}, \mathrm{Partial-Tally}, \mathrm{Recover})$.
\end{remark}
Doing so would mathematically demonstrate that {\BS} and {\NM} are equivalent, thus significantly reducing the burden of proofs for {\BS}, as one property would suffice to achieve the other. \cite{smyth2021} achieved this for earlier versions of the presented definitions, albeit these did not generalise to a larger class of voting schemes.

We also note that in formalising \cite{votetech}'s definitions, a degree of generality will inherently be lost. \cite{votetech} presents informal definitions with the purpose of keeping them general such that they can apply to a wider class of voting systems. Analysis of such privacy notions calls for mathematical rigour that will inevitably constrain the voting scheme definitions. Our definitions of election schemes however, are kept purposefully as broad as possible. Future work can focus on relaxing these definitions, particularly by allowing for a simpler key generation output (symmetric encryption uses one key only), and updating algorithms $\mathrm{Vote, Partial-Tally, Recover}$ accordingly. 

Further work can also be carried out to consider distributed tallying. Both \cite{smyth2021}'s and our definitions do not take into account when an adversary has the power to subvert the tallying operation, thus placing a trust assumption on the benign behaviour of the tallying authority. Our voting scheme definition performs a partial tally operation, thus we can elaborate on this by constructing games for ballot secrecy and ballot independence that give the adversary access to the tallying operation, and tasks them with winning the games. We assume there exists at least one honest tallier, and if the adversary cannot win, we know ballot secrecy or ballot independence will hold despite the adversary subverting the tallying operation.

This research avenue, while complex, can be particularly interesting in the context of consensus mechanisms in Distributed Ledger Technologies (DLTs) that are fair, as a means to replace Proof of Work or Proof of Stake. These consensus mechanisms are resource intensive and inherently unfair (\cite{vukolic2015quest}, \cite{guerraoui2018unfairness}) and could instead be replaced with a voting scheme as a way to fairly approve transactions. Indeed, DLTs such as the IOTA foundation are already implementing voting schemes as a consensus mechanism (\cite{muller2020fast}). Provided the voting scheme has a low computational complexity, this could be introduced in DLTs where the tallying is decentralised and with newer definitions of {\BS} and {\NM}, our framework can be used to test if the voting scheme has provable guarantees of voter privacy.

It is important to note that the work we present has the purpose of enabling transparency of the privacy properties of voting schemes, but as briefly discussed in section \ref{intro}, privacy does not imply verifiability nor vice versa. In fact, verifiability ensures fairness whereas privacy ensures freedom, and achieving both simultaneously is challenging when they can be seemingly contradictory. A voting scheme can preserve ballot secrecy but output random election outcomes, there are no conditions on the voting scheme being \textit{sound} or \textit{complete} for it to satisfy {\BS} or {\NM}. Conversely, as observed in our case study, a voting scheme may be verifiable yet not preserve {\BS} or {\NM}, such as Helios. 

Verifiability and privacy notions are challenging. Agreeing upon a universal definition of ballot secrecy or ballot independence is not straightforward, but we hope to get one step closer to that goal. Most importantly, we highlight that both verifiability and privacy are necessary conditions in a fair and free democratic election process.
\chapter{Conclusion}

Numerous definitions for ballot secrecy exist in academia. Some have been demonstrated to be too weak and fail to identify privacy breaches, and others are too strong. These are further discussed in Chapter \ref{litreview}. \cite{smyth2021} presents a definition of ballot secrecy that does not rely on the trust assumption that ballots are recorded as cast. The definition assumes that an adversary exists that controls the ballot collection process. He demonstrates that this definition is equivalent to his definition of non-malleable ballots. Improving upon earlier work, Smyth constructs a generalised definition of ballot secrecy and non-malleability in a private draft (\cite{votetech}), that is used in this report.

His newer definitions generalise to a larger class of voting systems, and are not exclusive to first-past-the-post. They are appealing for a general audience, but lack the formality required for rigorous proofs. We take the definitions of ballot secrecy and ballot independence (non-malleability) he presents and formalise them into algorithmic definitions, following the game-based models of cryptography. Subsequently we mathematically prove that ballot secrecy implies ballot independence.

Using the results of this proof we construct an adversary against the voting scheme Helios in Python and show that a winning adversary against ballot independence suffices to compromise ballot secrecy.  The Python framework allows us to test any voting scheme that satisfies the definition in  \ref{es} by running the games against said voting scheme and verifying if an adversary can win.
We show that if a winning adversary $\mathcal{A}$ against the {\NM} Python game exists, then an adversary $\mathcal{B}$ against {\BS} can win the game by simulating $\mathcal{A}$'s challenger. \textbf{This result shows that a voting system with malleable ballots does not preserve ballot secrecy.}

\printbibliography

\end{document}